\documentclass[10pt]{article}
\usepackage{amsmath}
\usepackage{amsfonts}
\usepackage{amssymb}

\begin{document}

\setcounter{page}{1} 

\newcommand{\gneg}{\neg} % negation
\newcommand{\sqc}{\mbox{\small \raisebox{0.0cm}{$\bigtriangleup$}}} % sequential conjunction
\newcommand{\sqd}{\mbox{\small \raisebox{0.049cm}{$\bigtriangledown$}}} % sequential disjunction
\newcommand{\ada}{\mbox{\Large $\sqcap$}} % choice universal quantifier
\newcommand{\ade}{\mbox{\Large $\sqcup$}} % choice existential quantifier

\newtheorem{theorem}{Theorem}[section]
\newtheorem{definition}[theorem]{Definition}\rm
\newtheorem{lemma}[theorem]{Lemma}
\newtheorem{corollary}[theorem]{Corollary}
\newtheorem{fact}[theorem]{Fact}
\newtheorem{proposition}[theorem]{Proposition}

\newenvironment{proof}{\noindent\bf Proof. \rm}{\mbox{}\hfill $\square$\vspace*{3mm}}

\title{A PSPACE-Complete First Order Fragment of Computability Logic}
\author{Matthew S. Bauer\\\\Department of Computer Science, Villanova University}
\date{}
\maketitle

\begin{abstract}
In a recently launched research program for developing logic as a formal theory of (interactive) computability, several very interesting logics have been introduced and axiomatized. These fragments of the larger \emph{Computability Logic} aim not only to describe \emph{what} can be computed, but also provide a mechanism for extracting computational algorithms from proofs. Among the most expressive and fundamental of these is CL4, shown in \cite{FTCII} to be (constructively) sound and complete with respect to the underlying computational semantics. Furthermore, the $\forall, \exists$-free fragment of CL4 was shown to be decidable in polynomial space. The present work extends this result and proves that this fragment is, in fact, PSPACE-complete.  

\end{abstract}

\section{Introduction}
Computability Logic (CoL), introduced by Japaridze in \cite{ICLI, GS}, is a research program for developing logic as a formal theory of interactive computability. Formulas in it are understood as (interactive) computational problems and logical operators represent operations on such entities. The goal of this program is to define a systematic way to answer the question ``what can be computed?" within the confines of a formal logical system. Computational problems are modelled as logical formulas through the use of game semantics. Each problem is understood as a game played between a machine $\top$ and its environment $\perp$, and a problem is seen as computable if there exists a machine that has an algorithmic winning strategy in the corresponding game. The closest predecessors to CoL are Hintikka's game-theoretic semantics \cite{HIN} and Blass's game semantics for Linear Logic \cite{LLGS}.

In line with its semantics, CoL introduces a rich set of logical connectives. Among those relevant to this paper are the propositional connectives $\neg$ (negation), $\vee$ (parallel disjunction), $\wedge$ (parallel conjunction), $\sqcup$ (choice disjunction) and $\sqcap$ (choice conjunction) as well as the ``choice" quantifiers $\ade$ and $\ada$. Intuitively, $\neg$ corresponds to switching the roles of the players $\top$ and $\perp$ in game to which it is applied: the game $\neg$A is obtained from A by turning $\top$'s legal moves and wins into legal moves and wins for $\perp$ and vice versa. Playing $A \wedge B$ or $A \vee B$ means simultaneously playing both of the games $A$ and $B$. Here, $\top$ wins the game $A \vee B$ if it wins (at least) one of the games $A$ or $B$. Similarly, $\perp$ wins the game $A \wedge B$ if it wins one of $A$ or $B$. 

The connective $\sqcup$ ($\sqcap$) corresponds to the player $\top$ ($\perp$) making a selection between two games. For example, the game \emph{Chess} $\sqcup$ \emph{Checkers} is one in which the machine ($\top$) chooses to play either \emph{Chess} or \emph{Checkers}. Symmetrically, in the game \emph{Chess} $\sqcap$ \emph{Checkers} it is the environment ($\perp$) that makes the decision. Assuming a universe of discourse of $\{0,1,2, ...\}$, $\ade x A(x)$ can be defined as $A(0) \sqcup A(1) \sqcup A(2) ...$ and $\ada x A(x)$ as $A(0) \sqcap A(1) \sqcap A(2) ...$. A telling example here is the game $\ada x \ade y (y = f(x))$, which represents computing the function $f(x)$. In this game, $\perp$ is required to make the first move and specify a value $m$ for $x$. $\top$ will then be required to make a counter-move corresponding to selecting a value $n$ for $y$. After these two moves, the game will be brought down to one of the form $n=f(m)$, which is won by the machine if true and by the environment if false.    

It should now be clear that many compelling problems can be modelled as games using the operators of CoL, where traditional computational problems form but a small fraction of such games. As a final example, consider the game $\ada x (p(x) \sqcup \neg p(x))$, where $p$ is some predicate. This game represents the problem of deciding $p$. Here, $\perp$ must specify some value $m$ for $x$, after which, $\top$ must select the true disjunct, i.e. either $p(m)$ or $\neg p(m)$ in order to win. A predicate can be seen to be decidable if and only if $\top$ has a winning strategy for this game. 

Several very interesting fragments of the larger Computability Logic have been introduced and axiomatized in recent works. Among the most expressive and fundamental of these is CL4, shown in \cite{FTCII} to be (constructively) sound and complete with respect to the underlying computational semantics. This is in the traditional sense that every formula provable in CL4 is valid and every valid formula is provable. Informally, a formula in CoL is said to be \emph{valid} if $\top$ has an algorithmic winning strategy in the game represented by that formula, regardless of how its atoms are interpreted. The soundness result for CL4 holds in an even stronger sense in that given a CL4 proof of a formula, we can effectively construct a machine that has a winning strategy in the game represented by that formula. This is because each proof of a logical formula encodes a solution to the problem that formula represents. It is, in particular, this (constructive) soundness result that allows us to shed light on answering the question of ``how" a problem can be computed.   

The $\forall, \exists$-free\footnote{Although lacking these operators, CL4$^-$ remains first order due to the presence of another, constructive sort of quantifiers $\ade$ and $\ada$.} fragment of CL4, which we call CL4$^-$, was shown in \cite{FTCII} to be decidable in polynomial space. The present work extends this result and proves that CL4$^-$ is, in fact, PSPACE-complete. Although our result points to the underlying difficulty in deciding this logic, it should be noted that CL4$^-$ is among the few decidable first-order logics with a natural semantics. Another decidable first order logic, and the closest relative (in its overall logical spirit) of CL4$^-$ is Multiplicative-Additive Linear Logic (MALL)\footnote{With the multiplicative operators of the latter being similar (but not identical) to the parallel operators of the former, and additive operators being similar to the choice operators.}. As proven in \cite{MALL}, the decision problem for first order MALL\footnote{The language of MALL contains function symbols, whereas the present version of CL4 does not. This difference, however, is not what creates the gap in complexity. As shown in \cite{PTA}, CL4 has a straightforward sound and complete extension whose language does contain function symbols. It would not be hard to show that this extension still remains in PSPACE. What creates the difference in complexities is the difference between the $\exists$-introduction rule of Linear Logic and the corresponding $\ade$-choose rule of CoL. In the bottom-up version of these rules, the former allows $x$ to be replaced by any term, while the latter only allows $x$ to be replaced by a variable or a constant. (In CoL a constant is not the same as a 0-ary function symbol.)} is NEXPTIME-hard.  

The inherent value of CoL extends well beyond the systematic study of interactive computational problems as its constructive qualities have proven useful in numerous applications. One of the most interesting of these applications has been in the development of applied theories based on CoL. Each of these theories relying heavily on CL4 as an important fragment of their underlying logic. In \cite{CL}, a CoL-based formal theory of arithmetic was introduced. This theory was shown to be sound and extensionally complete with respect to polynomial time computability. That is in the sense that every theorem in this theory represents an interactive number-theoretic computational problem with a polynomial time solution and every such problem is captured by some theorem. Similar theories for polynomial space and beyond have also been elaborated on in \cite{CLII}, \cite{CLIII}. 

Additional applications of CoL include knowledge base systems and systems for planning and action. Declarative and logic programming languages are also potential candidates for this currently expanding, list. A discussion of these applications has been given by Japaridze in section 10 of \cite{GS}.

\section{The Logic CL4$^-$}\label{cl4}

In the interest of self containment, what follows will provide a basic introduction to the logic CL4$^-$. Its axiomatization and auxiliary concepts are identical to those of CL4, the only difference being that formulas of the language of CL4 may contain an additional sort $\forall, \exists$ of quantifiers termed ``blind quantifiers", alien to the language of CL4$^-$. Both systems are analytic, which implies that CL4$^-$ is a conservative fragment of CL4 and hence it inherits the earlier mentioned soundness and completeness of the latter. For a complete treatment of the full CL4, see \cite{FTCI}, \cite{FTCII}. Although not necessary, it may also be helpful to have some formal acquaintance with the semantics of CoL before proceeding with the remainder of this paper. For rigorous definitions of these concepts, including formal definitions of CoL games and validity, the reader is advised to consult \cite{GS}.  

The languae of CL4$^-$ is made up of the operators $\neg , \vee, \wedge, \sqcap, \sqcup, \ade, \ada$, the logical letters $\top$ and $\perp$ as well as two --- \textbf{elementary} and \textbf{general} --- sorts of non-logical letters. Each such non-logical letter has a fixed arity and there is assumed to be an infinite amount of $n$-ary letters of either sort for every $n \geq 0$. Semantically, elementary letters represent what are called \emph{elementary games}. These are games in which no (legal) moves can be made, i.e. games that are automatically won or lost. As for general letters, they represent any, not necessarily elementary, games. We will use $p, q, r, s, ...$ and $P, Q, R, S, ...$ to denote elementary and general letters, respectively.  An \textbf{atom} of CL4$^-$ is defined as $L(t_1, ... , t_n)$ were L is any n-ary letter and $t_1, ..., t_n$ are terms, whereby a \textbf{term} is a variable $\{v_0, v_1, ... \}$ or a constant $\{0, 1, ... \}$. The \textbf{logical atoms} are $\top$ and $\perp$. Semantically, they are understood as (0-ary) elementary games won by $\top$ and $\perp$, respectively. 

\textbf{CL4$^-$-formulas} are constructed in the standard way by a combination of atoms and variables through the connectives $\neg , \vee, \wedge, \sqcap, \sqcup$ and the quantifiers $\ade, \ada$ with the restriction that $\neg$ can only be applied to (non-logical) atoms\footnote{Negated atoms do not count as atoms, so multiple negations are not allowed.}. This does not lead to any loss of expressiveness because, as is well known in CoL \cite{GS}, semantically De Morgan's laws (for either sort of conjunction and disjunction, as well as quantifiers) and double negation both hold. The definitions of free and bound occurrences of variables remain unchanged from classical first order logic and it should be understood that a variable in the scope of multiple quantifiers is bound to the one nearest it. For safety, we agree that, in what follows, ``formula" always means one where no variable has both free and bound occurrences. We say a CL4$^-$-formula is \textbf{general-base} if it does not contain any non-logical elementary letters. Likewise, a formula not containing any general letters is said to be \textbf{elementary-base}.     

A \textbf{positive occurrence} of an atom in a CL4$^-$-formula is an occurrence that is not in the scope of the $\neg$ operator. Otherwise the occurrence is \textbf{negative}. By a \textbf{surface occurrence} of a (sub)formula we mean a formula that is not is the scope of any choice connectives or quantifiers. A formula not containing any choice operators is said to be \textbf{choiceless} and if it additionally does not contain general letters it is then said to be \textbf{elementary}. Now the \textbf{elementarization} of a formula $F$ is the result of replacing in $F$ every surface occurrence of the form $G_1 \sqcap . . . \sqcap G_n$ or $\ada xG$ by $\top$, every surface occurrence of the form $G_1 \sqcup . . . \sqcup G_n$ or $\ade xG$ by $\perp$, every positive surface occurrence of each general atom by $\perp$, and every negative surface occurrence of each general atom by $\top$. Finally, a formula is said to be \textbf{stable} if and only if its elementarization is a valid formula of classical logic. It should be noted that the (non-choice) operators of CL4$^-$ behave exactly as in classical logic when they are contained within elementary formulas. 

\begin{definition}
With ${\cal P} \vdash F$ meaning ``from premise(s) ${\cal P}$ conclude $F$", CL4$^-$ is the system defined by the following rules of inference: 
\end{definition}
\textbf{Wait:} ${\vec{H}} \vdash F$, where $F$ is stable and $\vec{H}$ is a smallest set of formulas satisfying the following conditions: 
\begin{enumerate}
	\item[ $\bullet$] whenever $F$ has a surface occurrence of a subformula $G_1  \sqcap ... \sqcap G_n$, for each $i \in \{ 1,..., n\}$, $\vec{H}$ contains the result of replacing that occurrence     in $F$ by $G_i$;
	\item[ $\bullet$] whenever $F$ has a surface occurrence of a subformula $\ada x G(x)$, $\vec{H}$ contains the result of replacing that occurrence in $F$ by $G(y)$, where $y$ is a variable not occurring in $F$.
\end{enumerate}
\textbf{$\sqcup$-Choose:} $H \vdash F$, where $H$ is the result of replacing in $F$ a surface occurrence of a subformula $G_1 \sqcup ... \sqcup G_n$ by $G_i$ for some $i \in \{ 1,..., n\}$. 
\\\\
\textbf{$\ade$-Choose:} $H \vdash F$, where $H$ is the result of replacing in $F$ a surface occurrence of a subformula $\ade x G(x)$ by $G(t)$ for some term $t$ with no bound occurrences in $F$. 
\\\\
\textbf{Match:} $H \vdash F$, where $H$ is the result of replacing in $F$ two --- one positive and one negative --- surface occurrences of some $n$-ary general letter $P$ by an$n$-ary non-logical elementary letter $p$ that does not occur in $F$. 
\\\\
The preceding deductive system has no explicit set of axioms as stated, however, any formula derived by \emph{wait} from the empty set of premises can be viewed as such.  

\section{PSPACE Completeness}\label{proof}

In this section, we give our main result. Before doing so, let us turn our attention to the \emph{formula game}, a game played between two players (E and A) for a given fully quantified Boolean formula in prenex normal form. This game is covered in chapter 8 of \cite{Sipser}, but we shall describe it briefly here as well. A fully quantified Boolean formula in prenex normal form --- henceforth simply referred to as \textbf{QBF} --- is a formula in which all variables are quantified by either $\forall$ or $\exists$ and all such quantifiers appear at the beginning of the formula. In a play of the formula game for a given QBF $\phi$, players take turns removing the leftmost quantifier and corresponding variable in $\phi$ and replacing all occurrences of that variable by one of the truth values, ``1" (true) or ``0" (false). Here, Player E is responsible for removing all existential quantifiers and Player A is responsible for removing all universal quantifiers. The game ends when no quantifiers remain in the formula and the winner is determined by the truth of the resulting variable-free Boolean formula. If it evaluates to true, Player E is the winner. Otherwise, Player A has won. In complexity theory, this game is stated as a decision problem as follows, ``Does Player E have a winning strategy in the formula game for $\phi$?".  It is known \cite{Sipser} that this problem, being identical to ``Is $\phi$ true?", is PSPACE-complete. We employ this fact in our main theorem, constructing a polynomial time reduction from this problem to CL4$^-$-provability. Next we establish some conventions on the fully quantified Boolean formulas we deal with in the paper. 

We can safely assume, without loss of generality, that, in any QBF, the quantifiers strictly alternate, different occurrences of quantifiers bind different variables, and the first and last quantifiers are both $\exists$. In addition, we can also assume that the quantifier free section of the formula is in 3-cnf form. With any later reference to a QBF, it should be understood that this formula meets all of the preceding conditions. What follows is an auxiliary concept that will be employed in our PSPACE-completeness proof. 

\begin{definition}\label{strategyTree}
A \textbf{strategy tree} for a QBF $\phi$ is a tree that satisfies the following conditions:
\begin{enumerate}
	\item The tree has $n$ levels where $n$ is the number of quantifier occurrences in $\phi$.
	\item There is a single root node residing in level 1 of the tree. In odd levels of the tree, every non-leaf node should have exactly two children. Every non-leaf node in an even level of the tree should have exactly one child. 
	\item Every node in the tree is labeled with either a ``1" or an ``0". Furthermore, it is mandated that the nodes on every even level of the tree are given labels alternating between ``0" and ``1'' starting with the leftmost node labeled as ``0" and proceeding right. 
\end{enumerate}
\end{definition}

From our definition it should be clear that a strategy tree can be used to define a strategy for Player E in the formula game for $\phi$. Each path from root to leaf in the tree represents a play of the formula game in which the labels of odd level nodes are moves made by Player E and the labels of even level nodes are moves made by Player A. More formally, to define a play of the formula game for a particular path from root to leaf in a strategy tree,  let $x_1, x_2, ..., x_n$ be the sequence of labels on each of the nodes in this path, with $x_1$ being the label of the root node and $x_n$ being the label of the leaf node in that path. A play of the formula game for this path is then a traversal of this sequence from left to right in which the labels correspond to the moves by each player in the game. If the label is a ``1", then the play in the corresponding formula game is to remove the leftmost quantifier and variable and replace all occurrences of that variable by ``1" or ``true". Likewise if the label is an ``0", the play in the corresponding formula game is to remove the leftmost quantifier and variable and replace all occurrences of that variable by ``0" or ``false". We say that a strategy tree is winning, or simply refer to it as a \textbf{winning strategy tree} if and only if every path from root to leaf in the tree is a play of the formula game for $\phi$ won by Player E. With this definition in mind, the following lemma is straightforward. 

\begin{lemma}\label{strat}
Given any QBF $\phi$, Player E has a winning strategy in the formula game for $\phi$ iff there exists a winning strategy tree $\theta$ for $\phi$.
\end{lemma} 

The following technical lemma is also necessary. 

\begin{lemma}\label{stable} If $\Pi$ is a stable CL4$^-$-formula, $c$ is a constant, and $q$ is an elementary letter,  then the following CL4$^-$-formula $\psi$ is also stable: 

$$q(c) \vee ( \neg q(c) \wedge \Pi).$$
\end{lemma}
\begin{proof}
Let $\psi_e$ and $\Pi_e$ be the elementarizations of $\psi$ and $\Pi$ respectively. We need only show $\psi_e$ is valid in classical logic. Note that, since the letter $q$ is elementary, $ \psi_e = q(c) \vee ( \neg q(c) \wedge \Pi_e)$. The stability of $\Pi$ means that $\Pi_e$ is always true. Hence $\psi_e$ is simply equivalent to, $q(c) \vee \neg q(c)$ and, of course, is always true. 

\end{proof}

Here we give some additional terminology and conventions used in our proof. Let the \textbf{position} of a node in a tree be its distance from the leftmost node on the same level. Namely, on any level of a tree, the leftmost node is in position 0 and a node in position $n$ is $n$ places to the right of this node. The function QM($\omega$, $c$)  returns the formula that results from removing the leftmost quantifier and variable from the CL4$^-$-formula $\omega$ and replacing every occurrence of the quantified variable by the constant $c$. Finally, in what follows, we often simply say ``formula" in place of ``CL4$^-$-formula".  We are now ready to state our main result. Note that, in view of the soundness and completeness of CL4$^-$, the following theorem will hold with ``validity" in place of ``provability".

\begin{theorem}\label{completeness}
Deciding provability for CL4$^-$ is PSPACE-complete.
\end{theorem}
\begin{proof}
As noted earlier, determining validity (provability) of a CL4$^-$-formula was show in \cite{FTCII} to be decidable in polynomial space. To show that CL4$^-$-provability is PSPACE hard, we construct a polynomial time mapping reduction $f$ from formula game\footnote{The problem of existence of a winning strategy for Player E in the formula game.} to CL4$^-$-provability. The following steps are used to construct the CL4$^-$-formula $f(\phi)$ from an arbitary QBF $\phi$. Our construction maintains the record $\Theta$ initialized to $\phi$. At each step during the construction, $\Theta$ is updated to a smaller subformula of the original $\phi$ such that the changes made in all previous steps are not contained within the updated $\Theta$.     

\begin{enumerate}
	\item Let $\Omega, x$ be such that $\Theta = \exists x \Omega$. Replace $\Theta$ with the formula $\ade x\Omega $. Update $\Theta$ to $\Omega$. Repeat steps 2 and 3 until no quantifiers remain in $\Theta$.
	
	\item  Let $\Omega, x$ be such that $\Theta = \forall x \Omega$. Replace $\Theta$ with the formula $$(G(0) \sqcap G(1)) \vee \ade x(\neg G(x) \wedge \Omega)$$ where $G$ is any unary general letter not occurring elsewhere in $\Theta$. Update $\Theta$ to $\Omega$.
	
	\item Let $\Omega, x$ be such that $\Theta = \exists x \Omega$. Replace $\Theta$ by $\ade x\Omega$ and update $\Theta$ to $\Omega$.
	
	\item Now $\Theta$ is the quantifier free subformula of $\phi$. In it, replace every negative literal $\neg x$ by $(G(x) \vee \neg G(0))$ and every positive literal $x$ by $(G(x) \vee \neg G(1))$. In these replacements, G should be a unary general letter not occurring elsewhere in $\Theta$ and should be unique to each replacement. 

\end{enumerate}

For example, if $\phi$ is the QBF $\exists x \forall y \exists z ((\neg x \vee y \vee x ) \wedge (z \vee x \vee \neg z))$, then $f(\phi)$ is a CL4$^-$-formula that looks like $\ade x (P(0) \sqcap P(1)) \vee \ade y(\neg P(y) \wedge \ade z (((Q(x) \vee \neg Q(0)) \vee (R(y) \vee \neg R(1)) \vee (S(x) \vee \neg S(1)) ) \wedge ((T(z) \vee \neg T(1)) \vee (U(x) \vee \neg U(1)) \vee \neg (V(z) \vee \neg V(0))))$. 

It should be clear that our formula $f(\phi)$ can be constructed in polynomial time from the QBF $\phi$. What remains to show is that our function $f$ is in fact a mapping reduction. That is to say, Player E has a winning strategy in the formula game for $\phi$ if and only if $f(\phi)$ is provable in CL4$^-$. In the remainder of this proof, $\Theta$ is a CL4$^-$-formula, $G$ is a unary general letter and $g$ is from the elementary letters. Pick and fix an arbitrary QBF $\phi$.

``$\Rightarrow$" Assume Player E has a wining strategy in the formula game for $\phi$. By lemma \ref{strat} there exists a winning strategy tree, call it $\Delta$, for $\phi$. We use this tree to construct a CL4$^-$-proof of $f(\phi)$. This CL4$^-$-proof is represented as a tree of CL4$^-$-formulas, denoted $\Sigma$, where the root is the desired conclusion and each formula follows from its children by some rule of CL4$^-$. We usually see the nodes of this tree as the corresponding CL4$^-$-formulas, even though it should be remembered that the same formula may reside at different nodes. Let us also establish an association between the levels of $\Delta$ and the levels of $\Sigma$ that contain choice operators. To do this, we adjust the names of labels on the levels of $\Sigma$. Namely, let the root reside in level ``1"  of $\Sigma$ and let numerically increasing label names be placed on subsequent levels of the tree with the stipulation that even numbers are repeated three times. These three even levels are distinguished by placing a subscript on each number: ``t" on the top level, ``m" on the middle level and ``b" on the bottom level. For example, the names of the labels on the first four levels of $\Sigma$ would be $1$, $2_t$, $2_m$ and $2_b$. Once a level in $\Sigma$ is reached in which none of the formulas in that level contain choice operators, sequentially increasing numbers are placed as the names of labels of the remaining levels. Given this labeling scheme, any level in $\Sigma$ whose label contains the number $k$ (regardless of its subscript) is associated with the level in $\Delta$ that is also labeled as $k$. Intuitively, this association is done because, as will be seen shortly, moves by Player E (resp. A) in $\Delta$ will be ``mimicked" using one (resp. three) CL4$^-$ rule applications in $\Sigma$. All definitions previously established involving the levels of a tree are identical for these ``adjusted" levels. We will also often refer to a level (in $\Sigma$ or $\Delta$) simply by using its label.

The tree $\Sigma$ is constructed in a top down fashion, so initially place $f(\phi)$ at the root. Then we apply the following procedure. We will see shortly that it can be used to produce nodes in $\Sigma$ until no leaf nodes containing choice operators remain. Each step of the procedure deals with a formula in one of the following two forms: 

\begin{equation}\tag{3a} \Psi_1 ... \Psi_n (\ade x (\Theta))\Upsilon_1 ... \Upsilon_n \end{equation}
 
 \begin{equation}\tag{3b} \Psi_1 ... \Psi_n ((G(0) \sqcap G(1)) \vee \ade x (\neg G(x) \wedge \Theta))\Upsilon_1 ... \Upsilon_n \end{equation}
 where $n \in \mathbb{N}$, $x$ is some variable, every $\Upsilon_ i$ is a closing parenthesis, and every $\Psi_i$ is the (not well formed) expression $q_i(c_i) \vee ( \neg q_i(c_i) \wedge$ with $c_i \in \{0,1\}$ and $q_i$ being an elementary letter unique to each $i$. Formulas of the form (3a) will reside in odd levels of $\Sigma$ and formulas of the form (3b) will reside in levels $l_t$ for even $l$. We first observe that any formula of the form (3b) must be stable. This is so because, in the elementarization of 
 
\begin{equation}\tag{3c}
 ((G(0) \sqcap G(1)) \vee \ade x (\neg G(x) \wedge \Theta)),
\end{equation}the subformula $(G(0) \sqcap G(1))$ is replaced with
$\top$. This implies immediately that (3c) is stable. Now applying Lemma \ref{stable} $n$ times starting from formula (3c), we get that (3b) is indeed stable.

\emph{Procedure 1} - Let $\psi$ be a childless formula of the so-far constructed portion of $\Sigma$ containing choice operators. 

\emph{Case 1:} $\psi$ is of the from (3a). Let $l$ (for odd $l$) be the level of $\psi$ in $\Sigma$ and let $p$ be its position on that level. Further let $c$ be the label of the node in position $p$ on level $l$ in $\Delta$. Attach the formula QM($\psi$, $c$) as the child of $\psi$. Here $\psi$ follows from its child by $\ade$-\emph{choose}. 

\emph{Case 2:} $\psi$ is of the from (3b), residing in level $l_t$ (for even $l$) of $\Sigma$. Let $\psi_l$ be the formula that results from replacing the leftmost $\sqcap$-conjunction in $\psi$ by $G(0)$. Similarly let $\psi_r$ be the formula that results from replacing the leftmost $\sqcap$-conjunction in $\psi$ by $G(1)$. Attach $\psi_l$ and $\psi_r$ as the left and right children of $\psi$ respectively. As noted previously, $\psi$, residing in level $l_m$, is stable and hence it follows from its children by \emph{wait}. Further let $\psi_{LM} :=$ QM($\psi_l$, 0) be the child of $\psi_l$ and $\psi_{RM} :=$ QM($\psi_r$, 1) be the child of $\psi_r$. Both $\psi_l$ and $\psi_r$ follow from their respective children by $\ade$-\emph{choose}. Let the child of $\psi_{LM}$ be the formula that results from replacing, in $\psi_{LM}$, $G(0)$ by $g(0)$ and $\neg G(0)$ by $\neg g(0)$ for some elementary letter $g$ not occurring elsewhere in the formula. Similarly, let the child of $\psi_{RM}$ be the formula that results from replacing, in $\psi_{RM}$, $G(1)$ by $g(1)$ and $\neg G(1)$ by $\neg g(1)$. Both $\psi_{LM}$ and $\psi_{RM}$, residing in level $l_b$, follow from their respective children by \emph{match}.        

To see that \emph{Procedure 1} can be used to produce children for each formula containing choice operators in $\Sigma$, notice that, by our construction of $f(\phi)$, the root will be of the form (3a). Further, any formula of this form will have its child produced by \emph{case 1} of \emph{Procedure 1} and this child will take the form (3b). Formulas of this form will have their children produced by \emph{case 2} of \emph{Procedure 1} resulting in two nodes that are again of the form (3a). This cycle then repeats and thus \emph{Procedure 1} can be used to construct our tree into one in which every leaf node is a choiceless CL4$^-$-formula. Now the following procedure should be carried out on all leaf nodes containing two \textbf{literals} (atoms with or without negation) of the form $G(a)$ and $\neg G(b)$ for some general letter $G$ until no such leaf nodes remain.       

\emph{Procedure 2} - Let $\psi$ be a childless formula of the so far constructed portion of $\Sigma$ containing two literals of the form $G(a)$ and $\neg G(b)$. Attach, as the child of $\psi$, the formula that results from replacing in $\psi$ the literal $G(a)$ with $g(a)$ and the literal $\neg G(b)$ with $\neg g(b)$. Here $\psi$ follows from its child by \emph{match}.     

A this point, our construction of $\Sigma$ is complete. We now show that any leaf node $\beta$ can be justified by \emph{wait} from no premises. Notice, by the construction of $\Sigma$, $\beta$ will be of the form 

\begin{equation}\tag{3e}
\Psi_1 ... \Psi_n \Pi \Upsilon_1 ... \Upsilon_n
\end{equation}
where $n \in \mathbb{N}$, $\Pi$ is an elementary formula and both $\Upsilon_{i}$ and $\Psi_{i}$ are the same as in (3b). Notice that $\Pi$ is indeed elementary because \emph{Procedure 2} eliminates all general letters, as each such general letter in $f(\phi)$ always has exactly two --- one positive and one negative --- occurrences. Applying Lemma \ref{stable} $n$ times starting from $\Pi$, one can see that $\beta$ is stable as long as $\Pi$ is stable. Here stability simply means the same as classical validity as we deal only with elementary formulas. Thus, we need only show that $\Pi$ is classically valid, after which, it follows that $\beta$ is a consequence of \emph{wait} from no premises. 

Now $\Pi$ is what becomes of the quantifier free section of $\phi$ after applying to $\phi$ the function $f$ and carrying out some sequence of \emph{Procedures} \emph{1} and \emph{2} until no quantifiers remain and no applications of \emph{match} are possible. Thus $\Pi$ is a choiceless formula in conjunctive normal form. Recall from the definition of $f(\phi)$ that any positive literal $x$ in $\phi$ is replaced by $G(x) \vee \neg G(1)$ and any negative literal $\neg x$ is replaced by $G(x) \vee \neg G(0)$. Further observe that, by the construction of $\Sigma$, $\Pi$ will have each variable replaced by a constant 0 or 1 and these values will match the value in the other disjunct of $G(x) \vee \neg G(0)$ or $G(x) \vee \neg G(1)$ exactly when the variable $x$ is replaced by a value that makes the corresponding literal $x$ or $\neg x$ true in $\phi$. Each such subformula $G(a) \vee \neg G(a)$ will have its general letters reduced to elementary letters by \emph{Procedure 2}, after which it will become a tautology. That is, if selections by Player E and Player A in the formula game for $\phi$ make the resulting formula true, then $\Pi$ will be a tautology, because every true literal in such a variable-free Boolean formula is replaced by a tautology $g(a) \vee \neg g(a)$. But $\Sigma$ is constructed in such a way that the replacements of variables in $f(\phi)$ (and hence $\Pi$) correspond to some path from root to leaf in $\Delta$. Because $\Delta$ is a winning strategy tree, this sequence of moves in the formula game for $\phi$ makes its resulting variable-free Boolean formula true, and hence the formula $\Pi$ valid in classical logic.      

``$\Leftarrow$" 
Assume $f(\phi)$ is provable in CL4$^-$. This means there exists a proof tree, call it $\Sigma$, that proves $f(\phi)$. We use $\Sigma$ to construct a winning strategy tree $\Delta$ for $\phi$. Our previous conventions regarding the ``adjusted" labeling scheme for $\Sigma$ is also used here. 

We now establish some necessary facts about the formulas in $\Sigma$ that contain choice operators. Such formulas always take one of the following four forms:

\begin{equation}\tag{3e} \Psi_1 ... \Psi_n (\ade x (\Theta))\Upsilon_1 ... \Upsilon_n \end{equation}
 
 \begin{equation}\tag{3f} \Psi_1 ... \Psi_n ((G(0) \sqcap \neg G(1)) \vee \ade x (\neg G(x) \wedge \Theta))\Upsilon_1 ... \Upsilon_n \end{equation} 
 
 \begin{equation}\tag{3g} \Psi_1 ... \Psi_n ((G(c)) \vee \ade x (\neg G(x) \wedge \Theta))\Upsilon_1 ... \Upsilon_n \end{equation}
 
 \begin{equation}\tag{3h} \Psi_1 ... \Psi_n ((G(c)) \vee  (\neg G(c) \wedge \Theta))\Upsilon_1 ... \Upsilon_n \end{equation}
where (here and below) $n \in \mathbb{N}$, $c$ is a constant, $\Theta$ is a CL4$^-$-formula and both $\Upsilon_{i}$ and $\Psi_{i}$ are the same as in (3b). Furthermore, non-choiceless formulas on odd levels take the form (3e), non-choiceless formulas on even levels with a ``t" subscript take the form (3f), non-choiceless formulas on even levels with a ``m" subscript take the form (3g) and non-choiceless formulas on even levels with a ``b" subscript are of the form (3h). In the following, we justify the preceding claim.   

Notice that in $\Sigma$, a formula of the form (3e) resides at the root, i.e. in level 1, and the only possible CL4$^-$ justification rule that can be used to derive it is $\ade$-\emph{choose}. This is because there is a single surface occurrence of the choice quantifier $\ade$. Note that \emph{match} is not a possible derivation rule as there are no surface occurrences of general atoms and \emph{wait} is not a possible derivation rule as the formula is not stable. This formula must have a single child formula $\psi$ residing in the adjusted level $2_t$ and will be of the form (3f). The form (3f) is identical to (3b) and hence it is stable. Thus the possible derivation rules for $\psi$ are either \emph{wait} or $\ade$-\emph{choose}, as its only surface occurrences occur in a subformula of $\psi$ that looks like $\psi_s :=((G(0) \sqcap \neg G(1)) \vee \ade x (\neg G(x) \wedge \Theta))$.

We show that a derivation by $\ade$-\emph{choose} is not possible. Assume for a contradiction that $\psi$ follows from a single child formula by $\ade$-\emph{choose}. The form of this child must be identical to its parent with the exception that $\psi_s$ from $\psi$ is now a subformula of the form $((G(0) \sqcap \neg G(1)) \vee (\neg G(t) \wedge \Theta))$ for some term $t$. The derivation of this formula must then contain two children in which $\psi_s$ (and only $\psi_s$) is further changed to $((G(0)) \vee (\neg G(t) \wedge \Theta))$ and $((G(1)) \vee (\neg G(t) \wedge \Theta))$ respectively.  This is because an application of \emph{wait} that eliminates the choice operator $\sqcap$ must be used somewhere is this derivation. In at least one of these formulas, we have $t \neq 0$ or $t \neq 1$. It follows that a formula of the form

$$\Psi_1 ... \Psi_n ((G(c)) \vee  (\neg G(t) \wedge \Theta))\Upsilon_1 ... \Upsilon_n$$
for $t \neq c$, follows by \emph{wait} from no premises  after some sequence of rules are applied to it. Further, any sequence of rules applied to this formula cannot change its ($\wedge, \vee$)-structure and only \emph{match} can alter the literals $G(c)$ and $\neg G(t)$ after which they will become elementary. In fact, \emph{match} must indeed be applied here for otherwise $G(c)$ and $\neg G(t)$ would both be interpreted as $\perp$ in their elementarization and the formula would be underivable. Thus, for any formula matching the above form to be a conclusion of \emph{wait}, its derivation must at some point contain the subformula $((g(c)) \vee  (\neg g(t) \wedge \Theta))$. This subformula will be conjuncted with the remainder of the formula and hence it must follow by \emph{wait} from no premises, i.e. its elementarization must be valid in classical logic. Since $c$ and $t$ differ, there exists an interpretation $^*$ that interprets both literals $g(c)$ and $\neg g(t)$ as $\perp$. Under this interpretation, the overall formula cannot be true regardless of truth value for the elementarization of $\Theta$. This contradicts the fact that $f(\phi)$ is provable and hence $\psi$ cannot be derived by $\ade$-\emph{choose}. 

Having only one possible derivation rule, $\psi$ must follow by \emph{wait} and will have two children matching the form (3g) for $c=0$ and $c=1$ respectively. We can assume here that the subformula containing $G(0)$ is the left child and the subformula containing $G(1)$ is the right child. If this is not the case, then an equivalent proof can easily be constructed for which it is true. Each of these formulas resides in level $2_m$ and cannot be derived by \emph{match} or \emph{wait}. This is because formulas of the form (3g) are not stable and cannot contain both positive and negative surface occurrences of a general letter. However, each of these formulas has a single surface occurrence of the quantifier $\ade$, and hence both formulas must be derived by $\ade$-\emph{Choose}. This results in two additional children of the form (3h) that must reside in level $2_b$. Formulas that match the form (3h) cannot be derived by \emph{wait} as they are again not stable. This leaves $\ade$-\emph{Choose} and \emph{match} as the only two possible derivation rules for the formulas in level $2_b$. If the formulas are derived by $\ade$-\emph{Choose}, then \emph{match} must be the next derivation applied to the resulting formulas. Further, the order in which \emph{match} and $\ade$-\emph{Choose} are applied to these formulas is of no importance and we can therefore assume w.o.l.o.g. that \emph{match} is the next rule applied. This leaves level 3 of $\Sigma$ with two formulas matching the form (3e). The preceding sequence of derivations will repeat in every branch of the tree until an odd level contains a choiceless formula. Furthermore, every formula in this level will be choiceless and any formula below this level will also be choiceless. This concretely establishes our stipulations regarding the form of the formulas with choice operators on each level of $\Sigma$. 

Now let $\Delta$ be an $n$ level tree, where $n$ is the number of $\ade$ quantifiers in $f(\phi)$ such that $\Delta$ meets all the conditions of definition \ref{strategyTree} with the exception that nodes in odd levels are not yet labeled. To produce these labels, carry out the following procedure for each node residing in an odd level of $\Delta$.

\emph{Procedure 3} - Let $\eta$ be a node in position $m$ on level $l$ of $\Delta$. Further let $\psi$ be the formula in position $m$ on level $l$ of $\Sigma$. Because it resides in an odd level, $\psi$ must be of the form (3e) and as before can only be derived by $\ade$-\emph{choose} from a single child. If this child is the formula QM($\psi$,$1$), label $\eta$ with a ``1". Otherwise, label $\eta$ with a ``0".  

Clearly, $\Delta$ now meets all the conditions of definition \ref{strategyTree}. It remains only to see that $\Delta$ is a winning strategy tree. Let $y_1, y_2, ... , y_n$ be the labels of an arbitrary sequence of nodes from root to leaf in $\Delta$ and let $\beta$ be the formula that results from making these moves, i.e. selecting these constants in the formula game for $\phi$. Further, let $\{v_1, v_2, ... \}$ be the set of variables from $\phi$ to which ``1" was assigned in the preceding play of the formula game for $\phi$. For this set, there exits a path $z_1, z_{2_t}, z_{2_m}, z_{2_b},..., z_n$ in $\Sigma$ from the root to a leaf in an odd level $n$ such that the formula $\psi$ in the label of this node has a value of ``1" for all variables from the set $\{v_1, v_2, ... \}$. To the remainder of the variables in this formula will be assigned terms other than the constant ``1". Further, $\psi$ does not contain choice operators and it must be a consequence of \emph{wait} from no premises as $f(\phi)$ is provable in CL4$^-$. For simplicity, we may assume that $\psi$ does not contain general letters either, with all pairs $G(c)$ and $\neg G(d)$ replaced by $g(c)$ and $\neg g(d)$. So as a conclusion of \emph{wait}, $\psi$ is a classical valid formula. Observe now that $\beta$ is true. To see this, assume for a contraction that $\beta$ is false. For reasons similar to those pointed out in the (``$\Rightarrow$") part of our proof, every false literal of $\beta$ is replaced (in $\psi$) by $g(c) \vee \neg g(d)$ with $c \neq d$. Since $g$ occurs only in these two places, one can always select a model that makes all such subformulas $g(c) \vee \neg g(d)$ simultaneously false. With this model in mind, every false literal of $\beta$ is replaced by a false formula, and hence $\psi$ is false, contradicting the fact that $\psi$ is classically valid. That is, every path from root to leaf in $\Delta$ corresponds to a sequence of moves in the formula game for $\phi$ such that Player E is the winner and by definition this means that $\Delta$ is a winning strategy tree. To officially complete this direction, by lemma \ref{strat}, Player E has a winning strategy in the formula game for $\phi$.

\end{proof}

Before giving some additional facts, we reproduce the logic CL3 from \cite{FTCI}. This logic is the predecessor of CL4 and can be understood simply as its elementary-base fragment. The language and deductive machinery for CL3 is almost identical to that of CL4, with the restriction that it does not contain general letters and \emph{match} is not among its rules of inference. Again we use the notation CL3$^-$ to denote the $\forall, \exists$-free fragment of CL3. As shown in \cite{FTCI}, CL3$^-$ is also sound and complete (in the same sense as CL4) with respect to the semantics of CoL. Thus, the following two facts will also hold with ``validity" in place of ``provability".

\begin{fact}
Deciding provability for the general-base fragment of CL4$^-$ is PSPACE-complete. 
\end{fact}
\begin{proof}
In fact, our proof of the preceding theorem was a proof of the present fact, as the formula $f(\phi)$ is general base. 
\end{proof}

\begin{fact}
Deciding provability for CL3$^-$ is PSPACE-complete. 
\end{fact}
\begin{proof}
This follows directly from the proof of theorem \ref{completeness} with the following modifications. In the description of the function $f$, every reference to a general atom $G$ should be replaced by a reference to an elementary atom, $g$. Additionally, some straightforward modifications to both directions of the mapping reduction are required to eliminate references to general atoms and to remove the now unnecessary and invalid applications of the \emph{match} rule.   
\end{proof}

It is important to mention here that the language of the logic CL4$^-$ contains but a tiny fragment of the possible operators in Computability Logic. As a direct consequence of our result, deciding validity for more expressive CoL based languages containing the language of CL4$^-$ as a subset will be at least PSPACE-hard. Of course, any CoL based logic whose language contains the quantifiers of classical first order logic, such as CL4, is already undecidable. What is more interesting, and at this point unknown, is the complexity of deciding validity for the various non-trivial fragments of Computability Logic smaller than CL4$^-$. These include the logics CL1, CL2, CL5 and CL6. The first two, studied in \cite{PCLI, PCLII}, are nothing but the propositional fragments of CL3$^-$ and CL4$^-$ respectively. CL6, studied in \cite{ICC, CC}, is the choiceless fragment of CL2, and CL5, studied in \cite{ICC}, is the general-base fragment of CL6.

\end{document}